\documentclass[12pt]{article}

\usepackage[margin=1in]{geometry}
\usepackage{mathpazo}
\usepackage{stmaryrd}
\usepackage[english]{babel}
\usepackage[autostyle, english = american]{csquotes}

\usepackage{tikz}

\usetikzlibrary{
  arrows, 
  automata, 
  shapes, 
  shapes.geometric, 
  positioning,
  calc, 
  chains, 
  decorations.pathmorphing,
  intersections}

\usepackage{hyperref}
\hypersetup{pdfpagemode=UseNone}

\usepackage{amssymb, amsthm, amsmath, dsfont, mathtools}
\usepackage{latexsym}
\usepackage{eepic}
\usepackage{sectsty}
\usepackage{framed}

\newtheorem{theorem}{Theorem}
\newtheorem{lemma}[theorem]{Lemma}
\newtheorem{prop}[theorem]{Proposition}
\newtheorem{cor}[theorem]{Corollary}
\theoremstyle{definition}

\newcommand{\tinyspace}{\mspace{1mu}}

\newcommand{\rank}{\operatorname{rank}}
\renewcommand{\det}{\operatorname{Det}}

\newcommand{\abs}[1]{\lvert #1 \rvert}

\newcommand{\ip}[2]{\langle #1 , #2\rangle}

\newcommand{\biggip}[2]{\biggl\langle #1, #2 \biggr\rangle}

\newcommand{\floor}[1]{\lfloor #1 \rfloor}

\newcommand{\norm}[1]{\lVert\tinyspace #1 \tinyspace\rVert}

\newcommand{\iso}{\cong}

\newcommand{\I}{\mathds{1}}

\newcommand{\setft}[1]{\mathrm{#1}}

\newcommand{\Proj}{\setft{Proj}}

\newcommand{\Lin}{\setft{L}}

\newcommand{\pro}[1]{\setft{Prod}\left(#1 \right)}
\newcommand{\proS}[1]{\setft{Prod}\S\left(#1 \right)}
\newcommand{\corr}[1]{\setft{Cor}\left(#1 \right)}
\newcommand{\corrr}[2]{\setft{Cor}_{#1} \left( #2 \right)}
\newcommand{\diag}[1]{\setft{Diag}\left(#1 \right)}
\newcommand{\con}[1]{\overline{#1}}

\newcommand{\unitary}[1]{\setft{U}\left(#1\right)}

\newcommand{\sphere}[1]{\mathcal{S}\!\left(#1\right)}

\newcommand{\complex}{\mathbb{C}}

\newcommand\X{\mathcal{X}}
\newcommand\Y{\mathcal{Y}}

\renewcommand\S{\mathcal{S}}

\DeclareMathOperator{\spn}{span}

\begin{document}
\emergencystretch 3em

\title{\bf On decomposable correlation matrices}

\author{
  Benjamin Lovitz
  %\footnote{benjamin.lovitz@gmail.com}
  \\[2mm]
  {\it Institute for Quantum Computing and Department of Applied Mathematics}\\
  {\it University of Waterloo, Canada}}

\maketitle

\begin{abstract}
Correlation matrices (positive semidefinite matrices with ones on the diagonal) are of fundamental interest in quantum information theory. In this work we introduce and study the set of $r${\it-decomposable} correlation matrices: those that can be written as the Schur product of correlation matrices of rank at most $r$. We find that for all $r \geq 2$, every $(r+1) \times (r+1)$ correlation matrix is $r$-decomposable, and we construct ${(2r+1) \times (2r+1)}$ correlation matrices that are not $r$-decomposable. One question this leaves open is whether every $4 \times 4$ correlation matrix is $2$-decomposable, which we make partial progress toward resolving. We apply our results to an entanglement detection scenario.
\end{abstract}

%-----------------------------------------------------------------------------%
\section{Introduction}
Correlation matrices have been a topic of considerable interest in quantum information theory  \cite{Devetak2005,Cleve:2008:PPR:1391349.1391350, PhysRevA.81.062312,Haagerup2011,7389396,Yu2017BoundsOT, marwah2018characterisation}. This interest is due in part to Tsirelson's theorem~\cite{Tsirelson1987}, which reveals an intimate connection between correlation matrices and certain types of correlations that can arise from quantum systems. Another motivation is the identification of correlation matrices with Schur channels, examples of which include physically relevant channels such as generalized dephasing channels, cloning channels, and the Unruh channel \cite{PhysRevA.81.062312}. In this work, we introduce and study the set of {\it decomposable correlation matrices}, and apply our results to an entanglement detection scenario. We note that a related question regarding coherent states was recently studied in \cite{marwah2018characterisation}.

For a positive integer $n$, let $\Lin (\complex^n)$ denote the set of linear operators on $\complex^n$, and let $\corr{\complex^n}\subset \Lin(\complex^n)$ denote the set of {\it correlation matrices}: positive semidefinite matrices with diagonal entries all equal to one. We say a correlation matrix $P \in \corr{\complex^n}$ is $r${\it-decomposable} if it can be written as the Schur product $\odot$ (also known as the Hadamard product, entrywise product, or pointwise product) of correlation matrices of rank $\leq r$, i.e.,
\begin{align}\label{decomp}
P=R_1 \odot \dots \odot R_m
\end{align}
for some positive integer $m$ and correlation matrices $R_1,\dots,R_m \in \corr{\complex^n}$ with ${\rank(R_i)\leq r}$ for all $i \in \{1,\dots,m\}$. We use $\corrr{r}{\complex^n}$ to denote the set of $r$-decomposable matrices.

It is well known that $\corr{\complex^n}$ is a compact and convex set. To our knowledge, it is not known whether $\corrr{r}{\complex^n}$ is closed, and we leave this question unanswered. We show that $\corrr{r}{\complex^n}$ is not convex when $r \geq 1$ and $n \geq 2r+1$.

It is clear that $\corrr{r}{\complex^n}=\corr{\complex^n}$ for all $n \leq r$. We prove that ${\corrr{n-1}{\complex^n}=\corr{\complex^n}}$ for all $n \geq 3$, but $\corrr{r}{\complex^n} \subsetneq \corr{\complex^n}$ for all $n \geq 2r+1$. This leaves open the question of whether the containment  $\corrr{r}{\complex^n}\subseteq \corr{\complex^n}$ is strict for $n \in \{r+2, \dots, 2r\}$, and in particular whether $\corrr{2}{\complex^4}\subseteq \corr{\complex^4}$ is strict. We reduce the latter to a simpler question of whether every element of a certain subset of $\corr{\complex^4}$ can be written as the Schur product of just two rank-two correlation matrices, which could make the problem more tractable for analytical or numerical approaches.

We apply our results to the following entanglement detection scenario. Say we are given many copies of unknown pure states $v_1 v_1^*, \dots, v_n v_n^*$, on which we are allowed to perform any of the measurements $\{v_1 v_1^*, \I-v_1 v_1^*\}, \dots, \{v_n v_n^*, \I-v_n v_n^*\}$, and we wish to detect that for any partitioning of the space into subsystems of dimension $\leq r$, at least one of the states must be entangled. This scenario is similar to our $r$-decomposability question, as the only meaningful information to be gained from performing the allowed measurements is precisely the inner products $\ip{v_a v_a^*}{v_b v_b^*}$ for $a,b \in \{1,\dots, n\}$. In Proposition~\ref{counterexamplecor} we find cases of this scenario in which one can indeed detect entanglement.

In Section~\ref{mp} we review some mathematical preliminaries, in Section~\ref{correlation} we present our main results on $\corrr{r}{\complex^n}$ and apply them to entanglement detection, and in Section~\ref{4?} we study the question of whether the containment $\corrr{2}{\complex^4}\subseteq \corr{\complex^4}$ is strict.
%-----------------------------------------------------------------------------%
\section{Mathematical preliminaries}\label{mp}

Here we review some elementary facts and definitions we use. We often find it convenient to identify a complex Euclidean space by a symbol such as $\X$ or $\Y$, rather than specifying the isomorphic space $\complex^n$, because it allows us to refer to multiple spaces that could be isomorphic to each other. For any complex Euclidean space $\X$, let $\ip{\cdot}{\cdot}: \X \times \X \rightarrow \complex$ be the standard Euclidean inner product that is conjugate-linear in the first argument and linear in the second argument. Let $\norm{\cdot}=\sqrt{\ip{\cdot}{\cdot}}$ be the Euclidean norm, and define the set of {\it unit vectors} $\sphere{\X}$ as the set of vectors $x \in \X$ that satisfy $\norm{x}=1$. For a non-negative integer $a$ we let $e_a$ denote the standard basis vector with $1$ in the $a$-th position and zeros elsewhere. We use the convention $[m]:=\{1,\dots, m\}$ for any positive integer $m$.

 For a positive integer $m$ and complex Euclidean spaces $\mathcal{X}_1, \dots, \mathcal{X}_m$, we say a vector (or {\it tensor})
\begin{align}
x \in \mathcal{X}_1\otimes  \dots \otimes \mathcal{X}_m
\end{align}
is a {\it product vector} (or {\it elementary tensor}) if it is non-zero and can be written as
\begin{align}\label{product}
x=x_1 \otimes \dots \otimes x_m
\end{align}
for some collection of non-zero vectors $x_1 \in \X_1$, \dots, $x_m \in \X_m$. If $x$ is not a product vector and is non-zero then we say $x$ is {\it entangled}. We use $\pro{\X_1 : \dots : \X_m}$ to denote the set of product vectors in $\mathcal{X}_1\otimes  \dots \otimes \mathcal{X}_m$, and $\proS{\X_1 : \dots : \X_m}$ to denote the set of unit product vectors. We refer to the spaces $\X_1, \dots, \X_m$ that compose the space $\X_1 \otimes \dots \otimes \X_m$ as {\it subsystems}.

For positive integers $n$ and $m$, we frequently define sets of product vectors
\begin{align}
\{ x_a : a \in [n] \} \subset \pro{\X_1 : \dots : \X_m}
\end{align}
without explicitly defining for each $a \in [n]$ corresponding vectors $x_{a,1},\dots, x_{a,m}$ for which
\begin{align}
x_a=x_{a,1}\otimes \dots \otimes x_{a,m}.
\end{align}
In this case, we implicitly fix some such set of vectors $x_{a,1},\dots, x_{a,m}$ (they are unique up to scalar multiples $\alpha_{a,1} x_{a,1}, \dots, \alpha_{a,m} x_{a,m}$ such that $\alpha_{a,1}\cdots \alpha_{a,m}=1$), and refer to the vectors $x_{a,j}$ without further introduction. We use symbols like $a, b, c$ to index vectors, and symbols like $i, j, k$ to index subsystems.

We conclude this section by reviewing some elementary facts about correlation matrices. It is straightforward to verify that a matrix $P \in \Lin(\complex^n)$ is contained in $\corr{\complex^n}$ if and only if $P=T^* T$ for some linear operator $T \in \Lin(\complex^n, \complex^s)$ (and positive integer $s$), the columns of which form unit vectors. We say $P$ is {\it generated by} some set of unit vectors $\{v_a: a \in [n]  \}\subset \sphere{\complex^s}$ if these vectors can be chosen as the columns of $T$. Note that ${P(a,b)=\ip{v_a}{v_b}}$, so $P$ is the matrix of inner products (i.e. the {\it Gram matrix}) of any generating set of unit vectors. Two sets of unit vectors $\{v_a : a \in [n] \} \subset \sphere{\complex^{s_1}}$ and ${\{u_a : a \in [n] \} \subset \sphere{\complex^{s_2}}}$ with $s_1\leq s_2$ generate the same correlation matrix if and only if there exists an isometry $U \in \unitary{\complex^{s_1},\complex^{s_2}}$ such that $U v_a = u_a$ for all $a \in [n]$. This property follows from the standard result that two operators $T_1 \in \Lin(\complex^n,\complex^{s_1})$ and $T_2 \in \Lin(\complex^n,\complex^{s_2})$ satisfy $T_1^*T_1^{}=T_2^* T_2^{}$ if and only if $T_2=U T_1 $ for some isometry $U \in \unitary{\complex^{s_1},\complex^{s_2}}$. Note that by linearity, the linear dependence of every generating set is the same. It is straightforward to verify that a correlation matrix $P \in \corr{\complex^n}$ is $r$-decomposable if and only if there exists a positive integer $m$, complex Euclidean spaces $\X_1, \dots, \X_m$ with $\dim{\X_i} \leq r$ for all $i \in [m]$, and a set of unit product vectors ${\{x_a : a \in [n]\} \subset \proS{\X_1 :\dots : \X_m}}$ that generate $P$.

%-----------------------------------------------------------------
\section{Results on $r$-decomposable correlation matrices}\label{correlation}

Here we state and prove our main results on $r$-decomposable correlation matrices.

\begin{theorem}\label{rankreduce}
For any integers $r\geq 2$ and $n \leq r+1$, $\corrr{r}{\complex^n}=\corr{\complex^n}$. More generally, let $\X$ be a complex Euclidean space and $P \in \corr{\X}$ be a correlation matrix. If $\rank(P)\geq 3$ and $P$ is generated by a set of unit vectors that contains a vector linearly independent from the rest, then $P$ is $(\rank(P)-1)$-decomposable.
\end{theorem}
\begin{proof}
We first prove the general statement. Let $\{v_a: a\in [n]\}$ be a set of unit vectors that generate $P$ such that
\begin{align}
v_c \notin \spn \{ v_a : a \in [n]\setminus \{c\}\}
\end{align}
for some index $c \in [n]$.

If $v_{c}$ is orthogonal to every other vector, then the construction is easy: the set of vectors with each $v_a$ replaced by $v_a\otimes e_0$ for $a \neq c$, and $v_c$ replaced by $v'_c\otimes e_1$ for any unit vector $v'_c \in \spn \{v_a : a \in [n]\setminus \{c\}\}$ generates $P$. This is a $(\rank(P)-1)$-decomposition of $P$, since
\begin{align}
\dim \spn \{v_a : a \in [n]\setminus \{c\}\}=\rank(P)-1
\end{align}
and 
\begin{align}
\dim\spn\{e_0,e_1\}=2\leq \rank(P)-1.
\end{align}

If $v_{c}$ is not orthogonal to every other vector, then define
\begin{align}
\Pi:=\Proj \left( \spn \{ v_a : a \in [n]\setminus \{c\}\} \right),
\end{align}
and define two correlation matrices $R$ and $Q$ as
\begin{align}
R(a,b)=\frac{\ip{v_a}{\Pi v_b}}{\norm{ \Pi {v_a}} \norm{\Pi {v_b}}}
\end{align}
and
\begin{align}
Q(a,b)=\begin{cases}
\norm{\Pi {v_c}}, & a\neq b \text{ and } c\in \{a,b\}\\
1, & \text{otherwise.}
\end{cases}
\end{align}
It is straightforward to verify that $P=R\odot Q$. Indeed, for $c \notin \{a,b\}$,
\begin{align}
(R \odot Q) (a,b)=\frac{\ip{v_a}{\Pi v_b}}{\norm{ \Pi {v_a}} \norm{\Pi {v_b}}}=\ip{v_a}{v_b}.
\end{align}
Otherwise,
\begin{align}
(R \odot Q) (a,c)=\frac{\ip{v_a}{\Pi v_c}}{\norm{ \Pi {v_a}} \norm{\Pi {v_c}}}\norm{\Pi {v_c}}=\ip{v_a}{v_c},
\end{align}
and similarly, $(R \odot Q) (c,a)=\ip{v_c}{v_a}$. The correlation matrix $R$ has ${\rank(R)=\rank(P)-1}$, and is generated by the unit vectors $\Pi {v_a}/\norm{\Pi {v_a}}$. The correlation matrix $Q$ is clearly rank $2$. This completes the proof of the general statement.

For the first statement, let $r \geq 2$ be an integer. It is clear that $\corrr{r}{\complex^n}=\corr{\complex^n}$ for all $n \leq r$, and by the above construction, $\corrr{r}{\complex^{r+1}}=\corr{\complex^{r+1}}$.
\end{proof}

Now we find cases in which $\corrr{r}{\complex^n}\subsetneq \corr{\complex^{n}}$. We require the following lemma, which we reference without proof. We note that this lemma holds more generally over an arbitrary field.

\begin{lemma}[\cite{westwick1967,johnston2011characterizing}, Corollary 10 in \cite{tensor}]\label{l1} Let $m \geq 1$ be an integer, let $\mathcal{X}_1, \dots, \mathcal{X}_m$ be complex Euclidean spaces, and let $x_1, x_2 \in \pro{\mathcal{X}_1 : \dots : \mathcal{X}_m}$ be product vectors. Then the following statements are equivalent:
\begin{enumerate}
\item For all scalars $\alpha_1, \alpha_2 \in \complex$, it holds that $\alpha_1{x_1}+\alpha_2 {x_2}\in \pro{\mathcal{X}_1 : \dots : \mathcal{X}_m} \cup \{0\}$.
\item For some non-zero scalars $\alpha_1, \alpha_2 \in 
\complex \setminus \{0\}$, it holds that\\ ${\alpha_1{x_1}+\alpha_2 {x_2}\in \pro{\mathcal{X}_1 : \dots : \mathcal{X}_m}\cup \{0\}}$.
\item There exists at most a single index $j \in [m]$ for which $\dim\spn \{{x_{1,j}}, x_{2,j}\}=2$.
\end{enumerate}
\end{lemma}

\begin{theorem}\label{counterexample}
For all integers $r\geq 1$ and $n \geq 2r+1$, $\corrr{r}{\complex^n}\subsetneq \corr{\complex^{n}}$.
\end{theorem}

\begin{proof}
For $r=1$, the statement follows easily from the fact that the Schur product of any two rank-one correlation matrices is again rank one (see the proof of Lemma~\ref{closed}). Assume $r \geq 2$. We find a correlation matrix $P\in \corr{\complex^{2r+1}}$ that is not contained in $\corrr{r}{\complex^{2r+1}}$. This will prove the claim, as it implies that any correlation matrix in $\corr{\complex^{n}}$ with principal submatrix $P$ is not $r$-decomposable.

Let ${v_1},\dots,{v_{r+1}}$ be any linearly independent collection of unit vectors for which
\begin{align}\label{assumption}
\abs{\ip{v_a}{v_{a+2}}}>\abs{\ip{v_a}{v_{a+1}}}\cdot \abs{\ip{v_{a+1}}{v_{a+2}}}.
\end{align}
For example, one could choose any $p \in (0,1)$ and let $\ip{v_a}{v_b}=p$ for all ${a\neq b \in [r+1]}$. Let $\alpha_1,\dots, \alpha_r, \beta_2, \dots, \beta_{r+1} \in \complex \setminus \{0\}$ be any collection of non-zero scalars subject to the constraint that for all $a \in [r]$ it holds that ${\norm{\alpha_a v_a+\beta_{a+1} v_{a+1}}=1}$, and let $P$ be the correlation matrix generated by
\begin{align}
\{{v_1},\dots,{v_{r+1}},\alpha_1 {v_1}+\beta_2 {v_2},\alpha_2 {v_2}+\beta_3 {v_3}, \dots, \alpha_r v_{r}+\beta_{r+1}{v_{r+1}}\}.
\end{align}
Note that $\rank(P)=r+1$. For notational convenience, we extend the definition of ${v_a}$ to denote the $a$-th vector in this set for each $a\in [2r+1]$.
%
%We further restrict ${v_1},\dots,{v_{r+1}}$ to satisfy
%\begin{align}\label{assumption}
%\abs{\ip{v_a}{v_{a+2}}}>\abs{\ip{v_a}{v_{a+1}}}\cdot \abs{\ip{v_{a+1}}{v_{a+2}}}
%\end{align}
%for all $a \in [r-2]$ (see Example~\ref{counterexamplecor} for an example). 

We proceed by contradiction. The existence of an $r$-decomposition of $P$ is equivalent to the existence of a positive integer $m$, complex Euclidean spaces ${\mathcal{X}_1, \dots, \mathcal{X}_m \iso \complex^{r}}$, and unit product vectors $\{{x_a}: a \in [2r+1]\} \subset \proS{\X_1 : \dots : \X_m}$ such that ${\ip{x_a}{x_b}=\ip{v_a}{v_b}}$ for all $a,b\in[2r+1]$. By Lemma~\ref{l1}, this implies that for each $a \in [r]$,
\begin{align}\label{index}
\dim \spn \{{x_{a, i}}, {x_{a+1,i}}\}=2
\end{align}
for at most a single index $i \in [m]$. Furthermore, such an index indeed exists for every $a \in [r]$, since for all $a \in [r]$,
\begin{align}
\dim\spn\{{x_a},{x_{a+1}}\}=\dim\spn\{{v_a},{v_{a+1}}\}=2.
\end{align}
For each $a \in [r]$, fix $i_a \in [m]$ to denote the unique index that satisfies~\eqref{index}. Since $\abs{\ip{x_{a,i}}{x_{a+1,i}}}=1$ for all $i \neq i_a$, it must hold that $\abs{\ip{x_{a,i_a}}{x_{a+1,i_a}}}=\abs{\ip{v_{a}}{v_{a+1}}}$ for all $a \in [r]$. Note that
\begin{align}
\dim\spn \{ {x_{a,i}}: a \in [2r+1]\}\leq r
\end{align}
for all $i \in [m]$, and
\begin{align}
\dim\spn \{ {x_1},\dots,{x_{r+1}} \}=\dim\spn \{ {v_1},\dots,{v_{r+1}} \}=r+1,
\end{align}
so there must exist an index $a \in [r-2]$ such that $i_a \neq i_{a+1}$. Fix $a$ to denote one such index. Note that
\begin{align}
\abs{\ip{x_{a,i_a}}{x_{a+1,i_a}}}&=\abs{\ip{v_{a}}{v_{a+1}}}\\
\abs{\ip{x_{a+1,i_{a+1}}}{x_{a+2,i_{a+1}}}}&=\abs{\ip{v_{a+1}}{v_{a+2}}}\\
\abs{\ip{x_{a,i_{a+1}}}{x_{a+1,i_{a+1}}}}&=1\\
\abs{\ip{x_{a+1,i_{a}}}{x_{a+2,i_{a}}}}&=1,
\end{align}
from which it follows that
\begin{align}
\abs{\ip{x_{a,i_a}}{x_{a+2,i_a}}}&=\abs{\ip{v_{a}}{v_{a+1}}}\\
\abs{\ip{x_{a,i_{a+1}}}{x_{a+2,i_{a+1}}}}&=\abs{\ip{v_{a+1}}{v_{a+2}}},
\end{align}
but this implies
\begin{align}
\abs{\ip{v_{a}}{v_{a+2}}}&=\abs{\ip{x_{a}}{x_{a+2}}}\\
							&=\prod_{i=1}^m \abs{\ip{x_{a,i}}{x_{a+2,i}}}\\
							&\leq \abs{\ip{x_{a,i_a}}{x_{a+2,i_a}}} \cdot \abs{\ip{x_{a,i_{a+1}}}{x_{a+2,i_{a+1}}}}\\
							&=\abs{\ip{v_{a}}{v_{a+1}}}\cdot \abs{\ip{v_{a+1}}{v_{a+2}}},
\end{align}
a contradiction to~\eqref{assumption}. This completes the proof.
\end{proof}

\begin{cor}\label{notconvex}
For all integers $r\geq 1$ and $n \geq 2r+1$, $\corrr{r} {\complex^{n}}$ is not convex.
\end{cor}
\begin{proof}
We first prove that $\corrr{r}{\complex^{2r+1}}$ is not convex. Let $P \in \corr{\complex^{2r+1}} \setminus \corrr{r}{\complex^{2r+1}}$ be any correlation matrix constructed in Theorem~\ref{counterexample}. Since $\corr{\complex^{2r+1}}$ is contained in a real affine space of dimension $2r(2r+1)$, then by Carath{\'e}odory's theorem~\cite{rockafellar2015convex},
\begin{align}
P=\sum_{i=1}^s p(i) R_i
\end{align}
for some positive integer $s \leq 2r(2r+1)+1$, probability vector $p$, and extreme point correlation matrices $R_i$. By Corollary 2 in \cite{doi:10.1137/S0895479892240683}, $\rank(R_i) \leq \floor{\sqrt{2r+1}} \leq r$ for all $i \in [s]$. It follows that $\corrr{r}{\complex^{2r+1}}$ is not convex, since each $R_i$ is $r$-decomposable and $P$ is not.
%By Corollary 2 in \cite{doi:10.1137/S0895479892240683}, for any positive integer $n$, the extreme points of $\corr{\complex^{n}}$ have rank $\leq \floor{\sqrt{n}}$. For any integer $r \geq 2$, it holds that $\floor{\sqrt{2r+1}}\leq r$, so the extreme points of $\corr{\complex^{2r+1}}$ have rank $\leq r$, and thus are all $r$-decomposable.

For the general statement, let $n \geq 2r+1$ be any integer. For each $i \in [s]$, let ${R_i' \in \corr{\complex^n}}$ be any correlation matrix with $\rank(R_i')=\rank(R_i)\leq r$ that contains $R_i$ as the upper-left principal submatrix. Then
\begin{align}
P':=\sum_{i=1}^s p(i) R_i' \in \corr{\complex^n}
\end{align}
contains $P$ as the upper-left principal submatrix, so $P'$ is not $r$-decomposable. As before, it follows that $\corrr{r}{\complex^{n}}$ is not convex, since each $R'_i$ is $r$-decomposable and $P'$ is not.
\end{proof}

Now we apply our results to the following entanglement detection scenario. Say we are given many copies of unknown pure states $v_1 v_1^*, \dots, v_n v_n^*,$ with $v_1,\dots, v_n \in \S(\X)$ for an unknown complex Euclidean space $\X$. Suppose further that we are allowed to perform any of the measurements 
\begin{align}
\{v_1 v_1^*, \I-v_1 v_1^*\}, \dots, \{v_n v_n^*, \I-v_n v_n^*\}
\end{align}
on any of the states $v_1 v_1^*, \dots, v_n v_n^*$, and we wish to detect entanglement in the following sense. For some positive integer $r$, we wish to detect that for any complex Euclidean space $\X$, any set of unit vectors $v_1,\dots, v_n \in \S(\X)$ that are consistent with the measurement outcomes observed in the above scenario, and any decomposition $\X=\X_1 \otimes \dots \otimes \X_m$ of $\X$ into spaces of dimension $\dim(\X_i)\leq r$, at least one of the vectors  $v_1,\dots, v_n$ must be entangled.

In the above scenario, the only meaningful information that can be gained from the measurement outcomes is precisely the Gram matrix of $\{v_1 v_1^*, \dots, v_n v_n^*\}$ (the matrix of inner products $\ip{v_a v_a^*}{v_b v_b^*}$ for $a,b \in \{1,\dots, n\}$). Note that a correlation matrix $R$ is the Gram matrix of rank-one projectors if and only if $R=P \odot \overline{P}$ for some correlation matrix $P$. The above scenario is therefore equivalent to being given some correlation matrix $R$ that is the Gram matrix of rank-one projectors, and wishing to detect that for any correlation matrix $P$, if $R=P \odot \overline{P}$, then $P$ is not $r$-decomposable. In Proposition~\ref{counterexamplecor} we find examples of such entanglement detection.

\begin{prop}\label{counterexamplecor}
For any integer $r\geq 1$ and real number $0<p<1$, there exists a correlation matrix arising from a set of $2r+1$ unit vectors
\begin{align}
\{{v_1},\dots,{v_{r+1}},{v_{(1,2)}},{v_{(2,3)}},\dots,{v_{(r,r+1)}}\}
\end{align}
such that for all $a \neq b \in [r+1]$,
\begin{align}
\abs{\ip{v_a}{v_b}}^2=p^2,
\end{align}
and for all $a \in [r]$,
\begin{align}
\abs{\ip{v_a}{v_{(a,a+1)}}}^2=\abs{\ip{v_{a+1}}{v_{(a,a+1)}}}^2=\frac{1+p}{2}.
\end{align}
Furthermore, any such correlation matrix with $0<p<\frac{1}{r}$ is not $r$-decomposable.
\end{prop}

\begin{proof}
We first prove the existence of such a correlation matrix. The correlation matrix generated by the set of unit vectors
\begin{align}\label{unitvectors}
\left\{{v_1},\dots,{v_{r+1}},\frac{1}{\sqrt{2(1+p)}}({v_1}+{v_2}), \dots, \frac{1}{\sqrt{2(1+p)}}({v_{r}}+{v_{r+1}}) \right\},
\end{align}
with $\ip{v_a}{v_b}=p$ for all $a\neq b \in [r+1]$, satisfies the desired conditions. Indeed,
\begin{align}
\biggip{v_a}{\frac{1}{\sqrt{2(1+p)}}({v_a}+{v_{a+1}})}&=\frac{1}{\sqrt{2(1+p)}}(1+p)\\
&=\sqrt{\frac{1+p}{2}},
\end{align}
and similarly,
\begin{align}
\biggip{v_{a+1}}{\frac{1}{\sqrt{2(1+p)}}(v_a+v_{a+1})}=\sqrt{\frac{1+p}{2}}.
\end{align}

Now we prove that any such correlation matrix with $0<p<\frac{1}{r}$ is not $r$-decomposable. For $r=1$, the statement follows easily from the fact that the Schur product of any two rank-one correlation matrices is again rank one (see the proof of Lemma~\ref{closed}), and that for all $0<p<1$, any correlation matrix satisfying the conditions of the proposition has rank $\geq 2$. Assume $r \geq 2$. It is clear that
\begin{align}
\abs{\ip{v_a}{v_{a+2}}}>\abs{\ip{v_a}{v_{a+1}}}\cdot \abs{\ip{v_{a+1}}{v_{a+2}}}
\end{align}
for all $a \in [r-2]$. Thus, by the proof of Theorem~\ref{counterexample} it suffices to show that the vectors $\{{v_1},\dots,{v_{r+1}}\}$ are linearly independent, and that for all $a \in [r]$ it holds that ${{v_{(a,a+1)}}=\alpha_a {v_a}+ \beta_{a+1} {v_{a+1}}}$ for some non-zero scalars $\alpha_a, \beta_{a+1}\in \complex \setminus \{0\}$.

First, by Gershgorin's circle theorem~\cite{horn2013matrix}, the condition that $\abs{\ip{v_a}{v_b}}^2=p^2$ for all ${a \neq b \in [r+1]}$, along with $0<p<\frac{1}{r}$, implies that the vectors $\{{v_1},\dots,{v_{r+1}}\}$ are linearly independent. Second, for each $a \in [r]$ the principal submatrix of $P$ generated by the vectors $\{{v_a},{v_{a+1}},{v_{(a,a+1)}}\}$ is of the form
\begin{align}
P^{(a,a+1)}=
\begin{pmatrix}
1 & e^{i \phi_1} p & e^{i \phi_2} \sqrt{\frac{1+p}{2}} \\
e^{-i \phi_1} p & 1 & e^{i \phi_3} \sqrt{\frac{1+p}{2}} \\
e^{-i \phi_2} \sqrt{\frac{1+p}{2}} & e^{-i \phi_3} \sqrt{\frac{1+p}{2}} & 1
\end{pmatrix}
\end{align}
for some $\phi_1,\phi_2,\phi_3 \in [0,2\pi)$. Note that
\begin{align}
\det(P^{(a,a+1)})=p(1+p)(-1+\cos(\phi_1-\phi_2+\phi_3))\leq 0,
\end{align}
and since $P^{(a,a+1)}$ is positive semidefinite,
\begin{align}
\det(P^{(a,a+1)})=0.
\end{align}
This implies that $P^{(a,a+1)}$ has rank one or two. We can deduce $\rank(P^{(a,a+1)})\neq 1$ because $v_a$ and $v_{a+1}$ are linearly independent. Thus, $\rank(P^{(a,a+1)})= 2$, which implies ${{v_{(a,a+1)}}= \alpha_a{v_a}+ \beta_{a+1} {v_{a+1}}}$ for some scalars $\alpha_a, \beta_{a+1} \in \complex$, both of which must be non-zero because no entry in $P^{(a,a+1)}$ has unit magnitude.
\end{proof}

%---------------------------------------------------------------------
\section{Is the containment $\corrr{2}{\complex^4}\subseteq \corr{\complex^4}$ strict?}\label{4?}
Theorem~\ref{rankreduce} implies $\corrr{2}{\complex^3}=\corr{\complex^3}$, while Theorem~\ref{counterexample} implies $\corrr{2}{\complex^5}\subsetneq\corr{\complex^5}$. This leaves open the question of whether the containment ${\corrr{2}{\complex^4}\subseteq \corr{\complex^4}}$ is strict. For a correlation matrix $P \in \corr{\complex^4}$, it might seem possible that a $2$-decomposition~\eqref{decomp} exists only for large values of $m$, which could make our problem intractable. The following theorem allows us to restrict our attention to $m=2$.

\begin{theorem}\label{4equiv}
The following statements are equivalent:
\begin{enumerate}
\item $\corrr{2}{\complex^4}\subsetneq \corr{\complex^4}$.
\item There exists a correlation matrix $P \in \corr{\complex^4}$ such that $\rank (P)=3$, no vector in a generating set of $P$ is linearly independent from the rest, and $P$ is not $2$-decomposable into the Schur product of precisely two correlation matrices of rank $2$.
\end{enumerate}
\end{theorem}

Theorem~\ref{4equiv} shows that it suffices to consider rank-three correlation matrices for which no vector in a generating set is linearly independent from the rest. In Proposition~\ref{infinitefamily}, we construct $2$-decompositions of an infinite family of such correlation matrices, thus narrowing our question even further. We speculate that perhaps our construction can inspire a more general construction of all such correlation matrices.
\begin{prop}\label{infinitefamily}
Let $P \in \corr{\complex^4}$ be any correlation matrix generated by a set of unit vectors $\{v_a : a \in [4]\}$ such that there exists a real number $-1/2 < p <1$ for which $\ip{v_a}{v_b}=p$ for all $a\neq b \in [3]$, and there exist non-zero scalars ${\alpha_1, \alpha_2 \in \complex \setminus \{0\}}$ for which ${v_4=\alpha_1 (v_1+v_3)+\alpha_2 v_2}$. Then $P \in \corrr{2}{\complex^4}$.
\end{prop}

In the remainder of this section, we prove Theorem~\ref{4equiv} and Proposition~\ref{infinitefamily}. For Theorem~\ref{4equiv}, $(1 \Rightarrow 2)$ will follow from Lemma~\ref{lindep}, and $(2 \Rightarrow 1)$ will follow from Lemma~\ref{closed}. We now prove these lemmas.

\begin{lemma}\label{lindep}
For all integers $n\geq 3$ and $2 \leq r \leq n-1$, if $\hspace{.1em} \corrr{r}{\complex^n}\subsetneq \corr{\complex^n}$, then there exists a correlation matrix $P \in \corr{\complex^n}\setminus \corrr{r}{\complex^n}$ such that no vector in a generating set of $P$ is linearly independent from the rest.
\end{lemma}
\begin{proof}
By assumption, there exists $P \in \corr{\complex^n}$ that is not $r$-decomposable. If there exists a vector in a generating set of $P$ that is linearly independent from the rest, then by the proof of Theorem~\ref{rankreduce} there exists a decomposition $P=Q \odot R$ where $\rank(Q)=2$ and $\rank(R)= \rank(P)-1$. If there exists a vector in a generating set of $R$ that is linearly independent from the rest, then this process can be repeated until we have  a decomposition
\begin{align}
P=Q_1 \odot \dots \odot Q_m \odot R'
\end{align}
for which each $Q_i$ has rank $2$ and no vector in a generating set of $R'$ is linearly independent from the rest. Furthermore, $R'$ is not $r$-decomposable, for otherwise $P$ would be $r$-decomposable.
\end{proof}

To prove Lemma~\ref{closed}, we require the following lemma proven by the author in~\cite{tensor}. We note that this lemma holds more generally over an arbitrary field.
\begin{lemma}[Corollary 9 in \cite{tensor}]\label{linincor}
\sloppy
Let $n$ and $m$ be positive integers, let $\mathcal{X}_1, \dots, \mathcal{X}_m$ be complex Euclidean spaces, and let $\{{x_a}: a \in [n] \}\subset \pro{\mathcal{X}_1 : \dots : \mathcal{X}_m}$ be a set of linearly independent product vectors. If there exist non-zero scalars $\alpha_1, \dots, \alpha_n \in \complex \setminus \{0\}$ such that
\begin{align}
\sum_{a \in [n]}  \alpha_a {x_a}\in \pro{\mathcal{X}_1 : \dots : \mathcal{X}_m},
\end{align}
then the vectors $x_1,\dots, x_n$ are non-parallel in at most $n-1$ subsystems, i.e. ${\dim \spn \{ {x_{a, j}}: a \in [n]\} >1}$ for at most $n-1$ indices $j \in [m]$.
\end{lemma}
\begin{lemma}\label{closed}
For any integer $n \geq 3$, let $P \in \corr{\complex^n}$ be any correlation matrix of rank $n-1$ generated by a set of unit vectors $\{v_a : a \in [n]\}$ for which
\begin{align}
v_n= \sum_{a \in [n-1]} \alpha_a v_a
\end{align}
for some non-zero scalars $\alpha_1,\dots,\alpha_{n-1} \in \complex \setminus \{0\}$. For any integer $2\leq r \leq n-1$, if ${P \in \corrr{r}{\complex^n}}$, then $P$ is $r$-decomposable as the Schur product of $n-2$ correlation matrices of rank $\leq r$.
\end{lemma}
\begin{proof}
By assumption, there exists a positive integer $m\geq 2$, complex Euclidean spaces $\X_1,\dots,\X_m$, and unit product vectors $\{u_a : a \in [n]\} \subset \proS{\X_1 : \dots : \X_m}$ that generate $P$ and satisfy
\begin{align}
u_n = \sum_{a \in [n-1]} \alpha_a u_a,
\end{align}
where the vectors $\{ u_a : a \in [n-1]\}$ are linearly independent by the condition ${\rank(P)=n-1}$. By Lemma~\ref{linincor}, this implies $\dim\spn\{ u_{a,i} : a \in [n]\}>1$ for at most $n-2$ indices $i \in [m]$. For each $i \in [m]$, let $R_i$ be the correlation matrix generated by $\{ u_{a,i} : a \in [n]\}$, so that
\begin{align}\label{rdecomp}
P=R_1 \odot \dots \odot R_m.
\end{align}
Then $\rank{R_i}>1$ for at most $n-2$ indices $i \in [m]$.

We conclude by showing that for any correlation matrix $R$ and rank-one correlation matrix $R'$, $R \odot R'$ is a correlation matrix with $\rank(R \odot R')=\rank(R)$. This will complete the proof, since all the rank-one correlation matrices in the $r$-decomposition~\eqref{rdecomp} can be absorbed into the $\leq n-2$ correlation matrices of rank $>1$ to construct the desired decomposition.

It follows from Schur's product theorem that $R \odot R'$ is a correlation matrix \cite{horn2013matrix}. Since $R'$ is positive semidefinite and rank-one, then $R'=x x^*$ for some vector $x$. Furthermore, since $R'$ has ones on the diagonal, each element of $x$ has unit modulus. It follows that
\begin{align}
R \odot R'= R\odot x x^* = \diag{x} R \hspace{.3em} \diag{x}^*,
\end{align}
where $\diag{x}$ is the diagonal unitary matrix with $\diag{x}(a,a)=x(a)$. Since $\diag{x}$ has full rank, then $\rank(\diag{x} R \hspace{.3em} \diag{x}^*)=\rank(R)$, which completes the proof.
\end{proof}

Theorem~\ref{4equiv} follows easily from Lemma~\ref{lindep} and Lemma~\ref{closed}. Now we prove Proposition~\ref{infinitefamily}.
\begin{proof}[Proof of Proposition~\ref{infinitefamily}]
We have
\begin{align}
P=
\begin{bmatrix}
    1 & p & p & \alpha_1+(\alpha_1+\alpha_2)p  \\
    p & 1 & p & \alpha_2+2\alpha_1 p  \\
    p & p & 1 & \alpha_1+(\alpha_1+\alpha_2)p \\
    \con{\alpha_1}+(\con{\alpha_1}+\con{\alpha_2})p & \con{\alpha_2}+2\con{\alpha_1} p & \con{\alpha_1}+(\con{\alpha_1}+\con{\alpha_2})p & 1
\end{bmatrix}.
\end{align}
We construct $P$ as
\begin{align}
P=Q_1 \odot Q_2,
\end{align}
where
\begin{align}
Q_1=&
\begin{bsmallmatrix}
    1 & \sqrt{\frac{1+p}{2}} & p& \sqrt{\frac{1+p}{2}}  \\
    \sqrt{\frac{1+p}{2}} & 1 & \sqrt{\frac{1+p}{2}} & 1 \\
    p & \sqrt{\frac{1+p}{2}} & 1 & \sqrt{\frac{1+p}{2}} \\
    \sqrt{\frac{1+p}{2}} & 1 & \sqrt{\frac{1+p}{2}} & 1
\end{bsmallmatrix},
\\
Q_2=&
\begin{bsmallmatrix}
    1 & p\sqrt{\frac{2}{1+p}} & 1& \sqrt{\frac{2}{1+p}}(\alpha_1+(\alpha_1+\alpha_2)p)  \\
    p \sqrt{\frac{2}{1+p}} & 1 & p \sqrt{\frac{2}{1+p}} & \alpha_2+2\alpha_1 p  \\
    1 &p  \sqrt{\frac{2}{1+p}} & 1 & \sqrt{\frac{2}{1+p}}(\alpha_1+(\alpha_1+\alpha_2)p) \\
   \sqrt{\frac{2}{1+p}} (\con{\alpha_1}+(\con{\alpha_1}+\con{\alpha_2})p) & \con{\alpha_2}+2\con{\alpha_1} p & \sqrt{\frac{2}{1+p}}(\con{\alpha_1}+(\con{\alpha_1}+\con{\alpha_2})p) & 1
\end{bsmallmatrix}.
\end{align}

The equality is clear; it only remains to show that $Q_1$ and $Q_2$ are positive semidefinite and rank two.

First, it is easily verified that $Q_1$ is the correlation matrix generated by the unit vectors
\begin{align}
q_{1,1}&=e_0\\
q_{1,2}&=\sqrt{\frac{1+p}{2}}e_0+\sqrt{\frac{1-p}{2}}e_1\\
{q_{1,3}}&=pe_0+\sqrt{1-p^2}e_1\\
{q_{1,4}}&={q_{1,2}},
\end{align}
which implies $Q_1$ is positive semidefinite. Furthermore, $\rank(Q_1)\leq 2$, since these vectors span at most a two-dimensional space.

Second, we verify that $Q_2$ is the correlation matrix of the unit vectors
\begin{align}
{q_{2,1}}&= p \sqrt{\frac{2}{1+p}} e_{0}+\frac{\con{\alpha_1}}{\abs{\alpha_1}}\sqrt{\frac{1+p-2p^2}{1+p}}e_{1}\\
{q_{2,2}}&=e_{0}\\
{q_{2,3}}&={q_{2,1}}\\
{q_{2,4}}&=(\alpha_2+2\alpha_1 x)e_{0}+\abs{\alpha_1}\sqrt{2(1+p-2p^2)}e_{1},
\end{align}
which will complete the proof, since it implies $\rank(Q_2)\leq 2$ as above. The vectors $q_{2,1},q_{2,2}, q_{2,3}$ are easily seen to be normalized. For $q_{2,4}$, recall the normalization condition on $v_4$
\begin{align}
\ip{v_4}{v_4}=2\abs{\alpha_1}^2 (p+1) + \abs{\alpha_2}^2 + (\con{\alpha_1}\alpha_2+\alpha_1 \con{\alpha_2})2p=1,
\end{align}
which implies
\begin{align}
1-\abs{\alpha_2+2\alpha_1 p}^2=1- \left( \abs{\alpha_2}^2 + 4 \abs{\alpha_1}^2 p^2 + (\con{\alpha_1}\alpha_2+\alpha_1 \con{\alpha_2})2p\right)=2\abs{\alpha_1}^2 (1+p-2p^2).
\end{align}
It follows that ${q_{2,4}}$ is normalized. Now we show that the inner products between $q_{2,1},\dots,q_{2,4}$ reproduce $Q_2$. All are easily seen except $\ip{q_{2,1}}{q_{2,4}}$, which we now verify:

\begin{align}
\ip{q_{2,1}}{q_{2,4}}&=p \sqrt{\frac{2}{1+p}}(\alpha_2+2\alpha_1 p)+\frac{{\alpha_1}}{\abs{\alpha_1}}\sqrt{\frac{1+p-2p^2}{1+p}}\abs{\alpha_1}\sqrt{2(1+p-2p^2)}\\
							&= \sqrt{\frac{2}{1+p}}\left(p(\alpha_2+2\alpha_1 p)+\alpha_1 (1+p-2p^2)\right)\\
							&=\sqrt{\frac{2}{1+p}}(\alpha_1+(\alpha_1+\alpha_2)p).
\end{align}
This completes the proof.
\end{proof}

%-----------------------------------------------------------------------------%
\section{Acknowledgments}
I thank Norbert L{\"u}tkenhaus for first proposing this research topic, and for many helpful discussions in earlier stages of this work. I thank John Watrous for helpful discussions and comments on drafts of this manuscript. I thank Chi-Kwong Li, Ashutosh Marwah, and Daniel Puzzuoli for helpful discussions.

%-----------------------------------------------------------------------------%
\bibliographystyle{alpha}
\bibliography{hakye}
%-----------------------------------------------------------------------------%

\end{document}